\documentclass{article}
\usepackage[utf8]{inputenc}
\usepackage{amsmath}
\usepackage{amsthm}
\usepackage{tikz}
\usetikzlibrary{calc}

\newtheorem{theorem}{Theorem}
\newtheorem{lemma}[theorem]{Lemma}

\theoremstyle{definition}
\newtheorem{definition}[theorem]{Definition}

\DeclareMathOperator{\Diag}{Diag}
\newcommand{\ket}[1]{\mathinner{|{#1}\rangle}}

\definecolor{dgreen}{rgb}{0,0.6,0}

\newcommand{\bubble}{
\raisebox{-4mm}{
  \begin{tikzpicture}[x=1mm, y=1mm, line width=0.3mm]
  \draw (5, 5) circle (5);
  \end{tikzpicture}
  }
}

\newcommand{\crossing}{ 
  \raisebox{-4mm}{
  \begin{tikzpicture}[x=1mm, y=1mm, line width=0.3mm]
    \draw (0,0) -- (10,10);
    \draw (10,0) -- (6,4);
    \draw (4,6) -- (0,10);
  \end{tikzpicture}
  }
}

\newcommand{\cupcap}{ 
  \raisebox{-4mm}{
  \begin{tikzpicture}[x=1mm, y=1mm, line width=0.3mm]
    \draw (0,10) to [out=320, in=180] (5,7) to [out=0, in=220] (10,10);
    \draw (0,0) to [out=40, in=180] (5,3) to [out=0, in=140] (10,0);
  \end{tikzpicture}
  }
}

\newcommand{\cupIcap}{ 
  \raisebox{-4mm}{
  \begin{tikzpicture}[x=1mm, y=1mm, line width=0.3mm]
    \draw (0,10) to [out=320, in=180] (5,7)
      to [out=0, in=220] (10,10);
    \draw (0,0) to [out=40, in=180] (5,3) to [out=0, in=140] (10,0);
    \draw (5,3) -- (5,7);
  \end{tikzpicture}
  }
}

\newcommand{\idtwo}{ 
  \raisebox{-4mm}{
  \begin{tikzpicture}[x=1mm, y=1mm, line width=0.3mm]
    \draw (0,0) to [out=50, in=320] (0,10);
    \draw (10,0) to [out=130, in=220] (10,10);
  \end{tikzpicture}
  }
}

\newcommand{\ketzero}{ 
  \begin{tikzpicture}[x=1mm, y=1mm, line width=0.3mm]
    \draw (0,8) to [out=280, in=180] (2,4) to [out=0, in=260] (4,8);
    \draw (6,8) to [out=280, in=180] (8,4) to [out=0, in=260] (10,8);
  \end{tikzpicture}
}

\newcommand{\ketone}{ 
  \raisebox{-2mm}{
  \begin{tikzpicture}[x=1mm, y=1mm, line width=0.3mm]
    \draw (0,8) to [out=280, in=180] (2,4) to [out=0, in=260] (4,8);
    \draw (6,8) to [out=280, in=180] (8,4) to [out=0, in=260] (10,8);
    \draw (2,4) to [out=280, in=260] (8,4);
  \end{tikzpicture}
  }
}

\newcommand{\measureone}{ 
  \raisebox{-4mm}{
  \begin{tikzpicture}[x=1mm, y=1mm, line width=0.3mm]
    \draw[dgreen, line width = 0.4mm]
      (6,7)arc(75:-255:4 and 2);
    \draw (5,0) -- (5,2);
    \draw (5,4) -- (5,10);
  \end{tikzpicture}
  }
}

\newcommand{\measurezero}{ 
  \raisebox{-1mm}{
  \begin{tikzpicture}[x=1mm, y=1mm, line width=0.3mm]
    \draw[dgreen, line width = 0.4mm] (5,5) ellipse (4 and 2);
  \end{tikzpicture}
  }
}

\newcommand{\jumpingjack}{
  \begin{tikzpicture}[x=0.18mm, y=0.18mm, line width=0.3mm]
    \draw (25,0) -- (25,83);
    \draw (25,93) -- (25,111);
    \draw (36,0) -- (36,83);
    \draw (36,111) -- (36,93);
    \draw (58,0) .. controls (58,20)
      and (60,47) .. (86,47) .. controls (115,47)
      and (115,18) .. (115,0);
    \draw (115,111) .. controls (115,94)
      and (115,61) .. (86,61) ..  controls (65,61)
      and (58,79) ..  (58,83);
    \draw (58,93) -- (58,111);
    \draw (68,0) .. controls (68,18)
      and (72,36) ..  (86,36) .. controls (101,36)
      and (104,15) .. (104,0);
    \draw (68,111) -- (68,93);
    \draw (68,83) .. controls (68,79)
      and (79,72) .. (86,72) .. controls (104,72)
      and (104,98) .. (104,111);
    \draw (86,61) -- (86,47);
    \draw (137,0) -- (137,111);
    \draw (148,0) -- (148,111);
    \draw[dgreen, line width = 0.4mm] ($(47,95) +(52:47 and 8)$) arc (52:-232:47 and 8);
    \draw[dgreen, line width = 0.4mm] ($(47,95) +(83:47 and 8)$) arc (83:97:47 and 8);
  \end{tikzpicture}
}

\newcommand{\littleancilla}{
  \begin{tikzpicture}[x=0.26mm, y=0.26mm, line width=0.3mm]
    \draw (18,76) -- (18,60);
    \draw (18,53) -- (18,0);
    \draw (29,76) -- (29,59);
    \draw (29,52) -- (29,0);
    \draw (43,76) -- (43,59);
    \draw (43,52) -- (43,35);
    \draw (43,28) -- (43,0);
    \draw (54,76) -- (54,60);
    \draw (54,53) -- (54,34);
    \draw (54,27) -- (54,0);
    \draw (86,76) -- (86,32);
    \draw (86,24) -- (86,0);
    \draw (97,76) -- (97,32);
    \draw (97,24) -- (97,0);
    \draw (155,0) -- (155,76);
    \draw (166,0) -- (165,76);
    \draw (115,33) .. controls (115,65)
      and (148,61) .. (148,58) ..  controls (148,51)
      and (140,51) ..  (144,47);
    \draw (122,36) .. controls (122,54) and (140,58) .. (144,54);
    \draw (148,51) .. controls (151,47) and (140,43) .. (144,40);
    \draw (148,43) .. controls (151,40) and (140,36) .. (144,33);
    \draw (148,36) .. controls (151,33) and (140,29) .. (144,25);
    \draw (148,29) .. controls (155,22) and (122,15) .. (122,28);
    \draw (148,22) .. controls (153,16) and (115,4) .. (115,27);
    \draw[dgreen, line width = 0.4mm] ($(36,62) +(52:36 and 6)$) arc (52:-232:36 and 6);
    \draw[dgreen, line width = 0.4mm] ($(36,62) +(85:36 and 6)$) arc (85:95:36 and 6);
    \draw[dgreen, line width = 0.4mm] ($(83,36) + (25:48 and 8)$) arc (25:-205:48 and 8);
    \draw[dgreen, line width = 0.4mm] ($(83,36) + (54:48 and 8)$) arc (54:69:48 and 8);
    \draw[dgreen, line width = 0.4mm] ($(83,36) + (93:48 and 8)$) arc (93:119:48 and 8);
  \end{tikzpicture}
}

\newcommand{\threedancers}{
  \begin{tikzpicture}[x=0.18mm, y=0.18mm, line width=0.3mm]
    \draw (40,0) .. controls (40,75) and (29,11) .. (29,84);
    \draw (29,93) -- (29,111);
    \draw (29,0) .. controls (29,5) and (29,32) .. (32,39);
    \draw (36,47) .. controls (40,55) and (40,79) .. (40,83);
    \draw (40,92) -- (40,111);
    \draw (119,111) .. controls (119,36) and (61,43) .. (61,0);
    \draw (72,0) .. controls (72,36) and (108,36) .. (108,0);
    \draw (94,39) .. controls (104,29) and (119,14) .. (119,0);
    \draw (61,84) .. controls (61,72) and (79,54) .. (86,47);
    \draw (61,92) -- (61,111);
    \draw (72,85) .. controls (72,65) and (108,54) .. (108,111);
    \draw (72,93) -- (72,111);
    \draw (151,0) .. controls (151,79) and (140,3) .. (140,111);
    \draw (148,47) .. controls (151,53) and (151,107) .. (151,111);
    \draw (140,0) .. controls (140,6) and (140,32) .. (144,39);
    \draw[dgreen, line width = 0.4mm] ($(50,95) + (52:42 and 7)$) arc (52:-232:42 and 7);
    \draw[dgreen, line width = 0.4mm] ($(50,95) + (83:42 and 7)$) arc (83:97:42 and 7);
  \end{tikzpicture}
}

\newcommand{\friends}{
  \begin{tikzpicture}[x=0.15mm, y=0.15mm, line width=0.3mm]
    \draw (25,130) -- (25,115);
    \draw (25,106) -- (25,0);
    \draw (36,130) -- (36,114);
    \draw (36,105) -- (36,0);
    \draw (54,115) -- (54,130);
    \draw (54,105) .. controls (54,97)
      and (61,86) .. (68,79);
    \draw (83,65) .. controls (90,58)
      and (94,60) .. (94,51) .. controls (94,45)
      and (54,11) ..  (54,0);
    \draw (65,115) -- (65,130);
    \draw (75,87) .. controls (68,94)
      and (65,101) .. (65,105);
    \draw (90,72) .. controls (105,60)
      and (105,58) .. (105,51) ..  controls (105,40)
      and (65,11) ..  (65,0);
    \draw (94,130) .. controls (94,83)
      and (54,60) .. (54,51) .. controls (54,40)
      and (58,39) .. (68,29);
    \draw (94,0) .. controls (94,6)
      and (90,8) .. (83,15);
    \draw (105,130) .. controls (105,79)
      and (65,60) .. (65,51) .. controls (65,45)
      and (69,43) .. (76,36);
    \draw (90,22) .. controls (97,15)
      and (105,11) .. (105,0);
    \draw (122,130) -- (122,0);
    \draw (133,130) -- (133,0);
    \draw[dgreen, line width = 0.4mm] ($(45,116) + (50:36 and 6)$) arc (50:-230:36 and 6);
    \draw[dgreen, line width = 0.4mm] ($(45,116) + (84:36 and 6)$) arc (84:96:36 and 6);
  \end{tikzpicture}
}

\begin{document}

\title{Realizing an exact entangling gate using Fibonacci anyons}
\author{Stephen Bigelow and Claire Levaillant}
\date{January 2018}
\maketitle

\begin{abstract}
  Fibonacci anyons are attractive
  for use in topological quantum computation
  because any unitary transformation of their state space
  can be approximated arbitrarily accurately by braiding.
  However there is no known braid that
  entangles two qubits
  without leaving the space spanned by the two qubits.
  In other words,
  there is no known ``leakage-free'' entangling gate made by braiding.
  In this paper,
  we provide a remedy to this problem
  by supplementing braiding with measurement operations
  in order to produce an exact controlled rotation gate on two qubits.
\end{abstract}

\section{Introduction}

The topological approach to quantum computation was first proposed by Alexei
Kitaev \cite{k03}. A quantum computer would store qubits in the state space of
a collection of non-Abelian quasi-particles. Fibonacci anyons are one of the
simplest such quasi-particles. They are complete for quantum computation in the
sense that any unitary operation on the state of a collection of Fibonacci
anyons can be approximated arbitrarily well by braiding \cite{flw02}.

In this paper, we use measurement of collective charge and fusion of pairs of
anyons, in addition to the usual braiding.
We are then able to avoid the problem of {\em leakage},
which is when the state of the anyons
leaves the space we use to encode qubits.

\begin{theorem} \label{MAIN}
  Using a certain encoding of qubits with Fibonacci anyons,
  any unitary operation on a collection of qubits
  can be approximated arbitrarily accurately and without leakage. 
\end{theorem}

Explicitly,
a qubit will be encoded in the state of four anyons
that have zero collective charge.
We will make a certain controlled rotation gate on two qubits.

\begin{theorem} \label{THEOREM}
  There is a protocol that exactly performs
  the controlled rotation gate $CR(2 \pi/5)$
  on a pair of qubits made of Fibonacci anyons.
  It uses braiding,
  measurement of collective charge,
  and fusion of pairs of anyons
\end{theorem}

The controlled rotation gate $CR(2 \pi/5)$ is capable of entangling two qubits.
Any entangling gate on two qubits, together with the ability to approximate
single qubit gates by braiding, is sufficient to approximate any quantum
computation on any number of qubits \cite{dnbt02}.
Thus Theorem \ref{THEOREM} implies Theorem \ref{MAIN}.

\section{Background}

We assume the reader is familiar with the basic terminology of anyonic systems \cite{w10}. We work with {\em Fibonacci anyons}. There is only one
kind of non-trivial quasi-particle, namely a Fibonacci anyon, which has quantum dimension
$$\phi = \frac{1 + \sqrt{5}}{2}.$$
We will denote this quasi-particle by $1$ and the trivial particle by $0$. The only non-trivial fusion rule is
$$1 \times 1 = 0 + 1,$$
meaning that two anyons can fuse to either a single anyon or to the vacuum. The only non-trivial $R$ and $F$ matrices are
$$
R =
\begin{pmatrix}
  e^{-i 4\pi/5} & 0 \\
  0 & e^{i 3\pi/5}
\end{pmatrix},
\qquad F =
\begin{pmatrix}
  \phi^{-1} & \phi^{-1/2} \\
  \phi^{-1/2} & -\phi^{-1}
\end{pmatrix}.
$$

We use diagrams to represent the dynamics of Fibonacci anyons. Time progresses upwards. A local maximum represents a pair of anyons fusing to the
vacuum. A local minimum represents a pair of anyons with trivial total charge being created out of the vacuum. A trivalent vertex can represent
two anyons fusing to a single anyon, or one anyon unfusing into two anyons. A measurement projects the state of a group of anyons to a state with
total charge either $0$ or $1$. We represent measurements that project to a trivial charge by a horizontal ellipse around the anyons being
measured. Diagrams are subject to the following relations.
\begin{eqnarray*}
  \bubble &=& \phi, \\
  \crossing &=& e^{3 \pi i/5} \idtwo + e^{-3 \pi i/5} \cupcap, \\
  \cupIcap &=& \phi^{1/2} \idtwo - \phi^{-1/2} \cupcap,
\end{eqnarray*}
 $$\measurezero = 1, \qquad \measureone = 0.$$

We store a qubit in the state of four anyons with trivial collective charge. We denote the two usual basis elements using ``ket'' notation as
follows.
$$ \ket{0} = \ketzero, \qquad \ket{1} = \ketone. $$
We represent the state of a qubit by a pair of complex numbers
$$(a,b) = a\ket{0} + b\ket{1}$$
and a gate on a single qubit by a two-by-two matrix.

The tensor product operation places qubits side by side. We abbreviate a tensor product of qubits as a single ket containing a sequence of
symbols. For example,
$$\ket{0} \otimes \ket{0} = \ket{00}.$$

We represent the state of a pair of qubits by a four-tuple of complex numbers
$$(a,b,c,d)
  = a \ket{00} + b \ket{01} + c \ket{10} + d \ket{11}.$$
We represent a gate on a pair of qubits by a four-by-four matrix. All four-by-four matrices in this paper will be diagonal. Important gates for us
will be the controlled rotation gates $CR(\theta)$, including the controlled $Z$ gate $CZ = CR(\pi)$.
$$
CR(\theta) = \Diag(1, 1, 1, e^{i \theta}), \qquad CZ = \Diag(1, 1, 1, -1).
$$
The state of a system is only defined up to multiplication by a non-zero scalar, and a gate is really a projective transformation.

\section{Preparing the ancilla}

An {\em ancilla} is a collection of anyons that have been put into a known state. This state is designed to be useful when the anyons participate in an operation on
input qubits. Ancillas can be prepared ahead of time.

In this section, we describe a procedure to prepare a collection of three qubits in a certain state. We use braiding, measurement of collective charge, and fusions of
pairs of anyons. In general, when we perform a measurement or fusion during a quantum computation, we need a ``recovery procedure'' in case the
outcome is not the desired one. When we prepare an ancilla, we can simply start the preparation again whenever an outcome is a failure. Thus, we
do not worry about recoveries.

First, we describe some qubits and gates that can be made by braiding alone. Let $\sigma_k$ denote the elementary braid that exchanges anyons $k$
and $k+1$ by a positive half twist.

\begin{lemma}
The following gates can be made by braiding.
  $$
  R =
  \begin{pmatrix}
    e^{-i 4\pi/5} & 0 \\
    0 & e^{i 3\pi/5}
  \end{pmatrix},
  \qquad
  Z =
  \begin{pmatrix}
    1 & 0 \\
    0 & -1
  \end{pmatrix}.
  $$
\end{lemma}

\begin{proof}
They are $R = \sigma_1$ and $Z = \sigma_1^5$.
\end{proof}

\begin{lemma}
  The following qubits can be made by braiding alone.
  $$\ket{\alpha} = (1, \phi^{1/2}),$$
  $$\ket{\beta} = (1, -\phi^{1/2}).$$
\end{lemma}

\begin{proof}
  They are $\ket{\alpha} = \sigma_1 \sigma_2 \ket{0}$
  and $\ket{\beta} = \sigma_1^{-4} \sigma_2 \ket{0}$.
\end{proof}

Note that $\ket{\alpha}$ is two cups, one nested inside the other.  It can be thought of as a sort of sideways identity diagram.  Similarly,
$\ket{\beta}$ is like a sideways identity diagram but with the addition of five half twists.

\begin{lemma}
  There is a protocol that has a non-zero probability
  of creating the Bell state $\ket{\Psi^+} = (0,1,1,0)$.
\end{lemma}

\begin{proof}
  Let $D_1$ and $D_2$ be the diagrams shown in Figure \ref{fig:noninj}.
  \begin{figure}[ht]
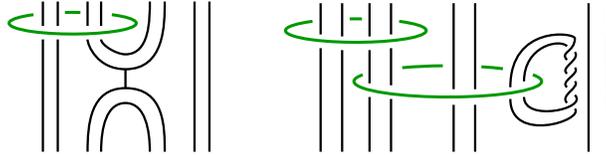

    \centering
    \jumpingjack \qquad \littleancilla
    \caption{Diagrams $D_1$ and $D_2$}
    \label{fig:noninj}
  \end{figure}
  By a diagrammatic calculation,
  these diagrams perform the following operations.
  $$
  D_1 = \Diag(0,1,1,\phi^{-1/2}),
  \qquad
  D_2 = \Diag(-\phi,1,1,0).
  $$
  The Bell state is then as follows.
  $$\ket{\Psi^+}
    = D_1 D_2 \ket{\alpha \alpha}.$$
\end{proof}

\begin{lemma}
  There is a protocol that has a non-zero probability
  of performing a Pauli-X gate
  $X = \begin{pmatrix} 0 & 1 \\ 1 & 0 \end{pmatrix}$.
\end{lemma}

\begin{proof}
  Place $\ket{\Psi^+}$ to the right of the input qubit.
  Take the adjacent pair of anyons
  from the input qubit and $\ket{\Psi^+}$
  and fuse them to the vacuum.
  Do this a total of four times
  until the input qubit has completely fused with
  the left qubit of $\ket{\Psi^+}$.
  The remaining right qubit of $\ket{\Psi^+}$
  is then the result of applying $X$ to the original input qubit.
\end{proof}

\begin{lemma}
  There is a protocol that has a non-zero probability
  of performing the CZ gate $\Diag(1,1,1,-1)$.
\end{lemma}

\begin{proof}
  Consider the two diagrams shown in Figure \ref{fig:nonunitary}.
  \begin{figure}[ht]
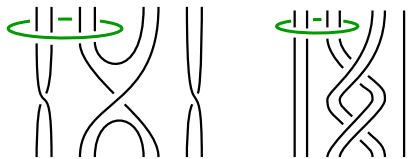

    \centering
    \threedancers \qquad \friends
    \caption{Diagrams $D_3$ and $D_4$}
    \label{fig:nonunitary}
  \end{figure}
  They perform the following two operations on a pair of qubits.
  $$
  D_3 = \Diag(1,1,1,-\phi^{-1}),
  \qquad
  D_4 = \Diag(1,1,1,-\phi^{-2}).
  $$

  We can use the $X$ gate to permute the entries of $D_4$.
  \begin{eqnarray*}
    (X \otimes X) D_4 (X \otimes X) &=& \Diag(-\phi^{-2},1,1,1) \\
    (X \otimes I) D_4 (X \otimes I) &=& \Diag(1,-\phi^{-2},1,1) \\
    (I \otimes X) D_4 (I \otimes X) &=& \Diag(1,1,-\phi^{-2},1)
  \end{eqnarray*}
  Compose the above three operations,
  and also $D_3^2$,
  to obtain the $CZ$ gate,
  up to an overall scalar.
\end{proof}

\begin{definition}
  We define the following three-qubit ancilla.
  $$\ket{\Gamma}
    = \phi^{1/2} \cos \frac{2 \pi}{5} \ket{\alpha 0 \alpha}
    - i \sin \frac{2 \pi}{5} \ket{\beta 1 \beta}$$
\end{definition}

\begin{lemma}
  There is a protocol that has a non-zero probability
  of producing the ancilla $\ket{\Gamma}$.
\end{lemma}

\begin{proof}
  Create the qubit
  $$\sigma_1^{2} \sigma_2^{-2} \ket{0} = \left(
      \phi^{1/2} \cos \frac{2 \pi}{5}, -i \sin \frac{2 \pi}{5}
    \right).$$
  Place that qubit in between two copies of the qubit $\ket{\alpha}$.
  Now perform a $CZ$ gate on the left two qubits,
  and then another $CZ$ gate on the right two qubits.
\end{proof}

\section{Performing the controlled rotation gate}

In this section, we describe how to perform the gate $CR(2 \pi/5)$ on a pair of qubits.
We will use an unlimited supply of ancillas $\ket{\Gamma}$
from the previous section.

A basic operation is to ``fuse'' two qubits into one, by annihilating two pairs of adjacent anyons, leaving two anyons from each of the original
qubits. The first stage of our protocol is to fuse the left and right qubits of $\ket{\Gamma}$ with the left and right input qubits of the gate.
Once this is done, the remaining qubit of the ancilla is used to apply one or the other of two diagonal entangling gates.

\begin{lemma}
  Suppose we are given two input qubits
  and an ancilla in the state $\ket{\Gamma}$.
  There is a protocol that,
  with non-zero probability,
  fuses the left and right input qubits
  with the left and right qubits of $\ket{\Gamma}$, respectively.
  If it fails,
  there is a recovery procedure
  to restore the input qubits to their original state.
\end{lemma}

\begin{proof}
  Position the ancilla $\ket{\Gamma}$
  between the left and right input qubits.
  In a successful application of the protocol,
  we perform four fusions of pairs of anyons,
  fusing to the vacuum in each case.

  The first fusion is of
  the adjacent anyons from the left input and the ancilla.
  We can make these fuse to the vacuum with probability one
  by the ``forced measurement'' procedure
  described in \cite{bfn08}.
  The idea is to alternate between
  measuring the collective charge of the two anyons we want to fuse,
  and of the four anyons that form the input qubit.
  With probability one,
  we will eventually measure the two anyons
  to have trivial collective charge.
  They are then guaranteed to fuse to the vacuum.

  The second fusion is of the adjacent anyons
  from what remains of the left input qubit
  and the left qubit of the ancilla.
  We want these to fuse to the vacuum.
  Suppose they instead fuse to a single unwanted anyon. Projection of the middle ancillary qubit, possibly followed by some application of the $Z$ gate, allows for a fusion with a guaranteed outcome aimed at retrieving an intact left input qubit. Explicitly,
  measure the middle ancillary qubit by fusing its left pair of anyons,
  thus projecting it to either $\ket{0}$ or $\ket{1}$.
  If it projects to $\ket{1}$ then
  apply five half twists to the fourth and fifth anyons from the left
  (either five clockwise
  or five counterclockwise half twists will work,
  since these have the same effect).
  Complete the recovery by fusing the third and fourth anyons,
  necessarily to $1$.
  We retrieve an intact left input qubit.
  Two of the qubits from $\ket{\Gamma}$
  have been destroyed in the process,
  and what remains of $\ket{\Gamma}$
  is not entangled with the input qubits, and can therefore be discarded.

  Suppose the left input qubit is fused with the left qubit of $\ket{\Gamma}$.
  The third fusion is of the adjacent anyons of the right input qubit
  and the ancilla.
  This is achieved by forced measurement, as before.

  The fourth and final fusion is of the adjacent anyons
  from what remains of the right input qubit
  and the right qubit of the ancilla.
  If they fuse to a single unwanted anyon,
  the recovery is similar to before.
  Project the untouched ancillary qubit
  to $\ket{0}$ or $\ket{1}$.
  If it projects to $\ket{1}$ then
  perform five half twists on the fourth and fifth anyons from the right,
  and also on the third and fourth anyons from the left.
  Complete the recovery by fusing the third and fourth anyons from the right,
  necessarily to a single anyon.
\end{proof}

Once we have successfully fused $\ket{\Gamma}$ into the two input qubits, the left and right qubits of $\ket{\Gamma}$ now carry the information
that was contained in the input qubits, and there is a remaining ancillary qubit in the middle. If we were to project that ancilla onto the
standard basis then we would randomly perform either the identity or $Z \otimes Z$ on the input qubits. We will instead project the ancilla onto a
different basis, which will randomly perform one of the following entangling gates.

\begin{definition}
  Let $G_1$ and $G_2$ be the following gates
  \begin{eqnarray*}
    G_1 &=& \Diag( z_1, \overline{z}_1, \overline{z}_1, z_1), \\
    G_2 &=& \Diag( z_2, \overline{z}_2, \overline{z}_2, z_2),
  \end{eqnarray*}
  where
  \begin{eqnarray*}
    z_1 &=& \cos \frac{2 \pi}{5} - i \sin \frac{2 \pi}{5}, \\
    z_2 &=& \phi \cos \frac{2 \pi}{5} + i \sin \frac{2 \pi}{5}.
  \end{eqnarray*}
\end{definition}

\begin{lemma} \label{lem:randomgate}
  Given two input qubits,
  there is a protocol that
  randomly performs either $G_1$ or $G_2$.
\end{lemma}

\begin{proof}
  Use the configuration from previous lemma, in which the left and right input
  qubits have been fused with the left and right qubits of $\ket{\Gamma}$.
  Fuse the middle two anyons of the remaining ancillary qubit.
  They can either fuse to the vacuum or to a single anyon.
  In either case,
  we can diagrammatically calculate the effect on the input qubits.
  If they fused to the vacuum then we have performed the gate $G_1$.
  If they fused to a single anyon then we have performed the gate $G_2$.
  Finally, dispose of the remaining two or three anyons from the middle ancilla.
\end{proof}

\begin{lemma} \label{lem:randominverse}
  Given two input qubits,
  there is a protocol that
  randomly performs either $G_1^{-1}$ or $G_2^{-1}$.
\end{lemma}

\begin{proof}
  The inverses of $G_1$ and $G_2$
  are their complex conjugates,
  up to an overall scalar.
  Follow the same protocol as in the previous lemma,
  but use the complex conjugate of $\ket{\Gamma}$.
  In order to create the complex conjugate of $\ket{\Gamma}$,
  simply apply a $Z$ gate to the middle qubit of $\ket{\Gamma}$.
\end{proof}

At this stage, we can randomly perform one of two entangling gates, and we can randomly perform one of their inverses. We use a random walk to
perform a specific entangling gate, determined ahead of time.

\begin{lemma}
  Given two input qubits,
  there is a protocol that performs the gate $G_1$
  with probability one.
\end{lemma}

\begin{proof}
  Apply the protocol from Lemma \ref{lem:randomgate}.
  If it performed $G_1$ then we are done.
  Suppose instead it performed the gate $G_2$.

  Now apply the protocol from Lemma \ref{lem:randominverse}.
  If it applied $G_2^{-1}$
  then the input qubits have been restored to their original state.
  If it applied $G_1^{-1}$
  then we have, in total, applied
  $G_1^{-1} G_2$ to the original qubits.

  Continue a process of applying
  Lemmas \ref{lem:randomgate} and \ref{lem:randominverse}.
  Note that diagonal matrices commute,
  so at any given stage
  we will have performed an operation of the form $G_1^k G_2^l$.
  Since $G_1$ has order five,
  we can assume $0 \le k < 5$.
  Also, let us choose our protocols
  to ensure that $l$ is always either $0$ or $1$.
  Specifically,
  if $l = 0$ then apply Lemma \ref{lem:randomgate},
  and if $l = 1$ then apply Lemma \ref{lem:randominverse}.

  The above protocol performs
  a random walk on a set of ten states.
  With probability one,
  it will eventually have performed the gate $G_1$.
\end{proof}

The gate $G_1$ is not exactly the controlled rotation gate we promised, but this is easily fixed.

\begin{proof}[Proof of Theorem \ref{THEOREM}]
  The gate $CR(\frac{2 \pi}{5})$
  is the braid $R^{-2} \otimes R^{-2}$
  composed with the gate $G_1$.
\end{proof}

\end{document}